\documentclass{IEEEtran}
\usepackage{cite}
\usepackage{amsmath,amssymb,amsfonts}
\usepackage{textcomp}
\def\BibTeX{{\rm B\kern-.05em{\sc i\kern-.025em b}\kern-.08em
    T\kern-.1667em\lower.7ex\hbox{E}\kern-.125emX}}

\usepackage{algorithmic}
\usepackage{graphicx}
\usepackage{textcomp}
\usepackage{subfigure}
\usepackage[ruled,vlined,norelsize]{algorithm2e}
\usepackage[justification=centering]{caption}
\usepackage{color}

\newtheorem{remark}{Remark}
\newtheorem{definition}{Definition}
\newtheorem{problem}{Problem}

\newtheorem{proposition}{Proposition}
\newtheorem{eg}{Example}

\newtheorem{lemma}{Lemma}

\newtheorem{theorem}{Theorem}

\begin{document}
\title{Event Concealment and Concealability Enforcement~in Discrete Event Systems Under~Partial~Observation}
\author{Wei Duan, Christoforos N. Hadjicostis,~\IEEEmembership{Fellow, IEEE}, and Zhiwu Li,~\IEEEmembership{Fellow, IEEE}	
	\thanks{This work was supported in part by the National R\&D Program of China under Grant No. 2018YFB1700104, the National Natural Science Foundation of China under Grants 61603285 and 61873342, and the Science Technology Development Fund, MSAR, under Grant 0012/2019/A1. (Corresponding author: Christoforos N. Hadjicostis.)} 
	\thanks{Wei Duan is with the School of Electro-Mechanical Engineering, Xidian University, Xi'an 710071, China (e-mail: dwei1024@126.com).}
	\thanks{Christoforos N. Hadjicostis is with the Department of Electrical and Computer Engineering, University of Cyprus, Nicosia 1678, Cyprus (e-mail: hadjicostis.christoforos@ucy.ac.cy).}
	\thanks{Zhiwu Li is with the School of Electro-Mechanical Engineering, Xidian University, Xi’an 710071, China, and also with the Institute of Systems Engineering, Macau University of Science and Technology, Taipa 999078, Macau SAR, China (e-mail: zhwli@xidian.edu.cn).}}

\maketitle

\begin{abstract}
	In this paper, we investigate event concealment and concealability enforcement in discrete event systems modeled as non-deterministic finite automata under partial observation. 
	Given a subset of secret events in a given system,
	{\em concealability} holds if the occurrences of all secret events remain hidden to a curious observer (an eavesdropper). 
	When concealability of a system does not hold, we analyze how a defensive function, placed at the interface of the system with the eavesdropper, can be used to enforce concealability by manipulating the observations generated by the system using event deletions, insertions, or replacements. 
	The defensive function is said to be $\emph{C}$-enforcing if, following the occurrence of secret events and regardless of earlier and subsequent activity generated by the system, it can always deploy a strategy to manipulate observations and conceal the secret events perpetually.
	We propose a polynomial complexity construction for obtaining one necessary and one sufficient condition for $\emph{C}$-enforceability.
\end{abstract}


\section{Introduction}
\label{sec:introduction}
\IEEEPARstart{S}{ecurity} and privacy are vital to emerging cyber-physical systems due to the increasing reliance of many applications on shared cyber infrastructures \cite{CPS2}. 
In the framework of discrete event systems (DESs) \cite{2010Lafortunebook}, \cite{2020Chrisbook}, where the underlying system is usually modeled as a finite automaton or a Petri net \cite{2017Ma}, \cite{2018Zhu},
problems of privacy protection have been widely investigated in terms of an {\em information flow} property called opacity \cite{2007Chris}, \cite{2013Wu}, \cite{2011Lin}.

Opacity characterizes the ability of a system to hide a secret from an external observer, i.e., a curious eavesdropper (or simply eavesdropper in the sequel) who is assumed to have full knowledge of the system's structure and partial observation capability.	
Opacity enforcement becomes necessary when a system is not opaque, and has been extensively studied mainly following two approaches: supervisory control strategies that restrict the system's behavior in order to avoid opacity violations \cite{2012Chrisa}, \cite{2010Dubreil}, \cite{2015YinB}, or enforcement via output obfuscation \cite{2014Wu}, \cite{2016Wu}, \cite{2018Wu}.

In most opacity research, the given secret is generally denoted as a subset of executions \cite{2011Lin} or a subset
of states \cite{2007Chris}, \cite{2013Wu}, \cite{2008Chris}.
In this paper, we focus on the case when certain events of a system are deemed secret and study (i) the circumstances under which their occurrences get revealed to an external eavesdropper, and (ii) ways to conceal (as long as needed) the secret events using an obfuscation mechanism.

To answer these questions, the privacy of a system is considered in terms of concealing secret events and the concept of event concealment is proposed.
The main contributions are as follows: (i) We systematically provide a formal description of secret events and the notion of concealability, make connections to existing notions of opacity, and briefly discuss a diagnoser-based approach that is developed to verify concealability of the system.
(ii) We consider the problem of concealability enforcement when the system is unconcealable, using a so-called defensive function placed at the output of the system to obfuscate the eavesdropper observations by appropriately modifying the original observations generated by the system (via event deletions, insertions, or replacements).
Taking advantage of the special structure of the concealability problem, we propose a verifier-like structure of polynomial complexity to obtain one necessary condition and one sufficient condition for enforceability of the defensive function with polynomial complexity.

\section{Preliminaries}

\subsection{System Model}

Let $E$ be a finite set of events and $E^*$ be the set of all finite-length strings (finite sequences of events) over $E$ including the empty string~$\varepsilon$. For a string $\lambda\in E^*$, we denote its length by $|\lambda|$ (the length of the empty string is defined as $|\varepsilon|=0$). A language $L\subseteq E^*$ is a subset of finite-length strings.
A prefix of a string $\lambda\in E^*$ is a string $u\in E^*$ for which there exists a string $v\in E^*$ such that $uv=\lambda$, where $uv$ denotes the concatenation of strings $u$ and $v$.
The prefix closure of $\lambda$ is the set of all of its prefixes denoted by $\overline{\lambda}=\lbrace u\in E^*|\exists v\in E^*, uv=\lambda \rbrace$.

We consider the event concealment problem in a DES modeled as a non-deterministic finite automaton (NFA) $G =(X, E, f, X_0)$, where $X$ is the finite set of states, $E$ is the finite set of events, $f : X\times E \rightarrow 2^X$ is the state transition function, and $X_0\subseteq X$ is the set of possible initial states. 
The transition function $f$ can be extended to $X\times E^*\rightarrow 2^X$ in the usual manner.
The generated language of $G$, denoted by $L(G)$, is defined as $L(G) =\lbrace \lambda\in E^*\mid\exists x\in X_0, f(x,\lambda)\ne\emptyset\rbrace$.
Without loss of generality, we assume that the set $X_0$ is a singleton (i.e., $X_0=\lbrace x_0\rbrace$).
Due to partial observation of the system, the event set $E$ can be partitioned into a set of observable events $E_o$ and a set of unobservable events $E_{uo}=E- E_o$.
The natural projection $P: E^*\rightarrow E_o^*$ strips out the unobservable events from a given sequence of events, and
is defined recursively for $\lambda\in E^*$ as $P(\lambda e) =P(\lambda)$ if $e\in E_{uo}$ and $P(\lambda e)=P(\lambda) e$ if $e\in E_o$ (we set $P(\varepsilon)=\varepsilon$ for $\varepsilon$).
The inverse projection of $P$ is defined as $P^{-1}: E_o^*\rightarrow 2^{E^*}$ with $P^{-1}(\omega)=\lbrace \lambda\in E^*|P(\lambda)=\omega\rbrace$ for $\omega\in E_o^*$.

\subsection{Diagnosability and Opacity}
Consider a system modeled as an NFA $G =(X, E, f, x_0)$ with $E=E_o\dot{\cup} E_{uo}$.
We briefly recall the notions of diagnosablity \cite{1995Sampath} and current-state opacity \cite{2007Chris}.
Let $E_f\subseteq E_{uo}$ denote the set of fault events, and $S\subseteq X$ denote the set of secret states. 
For simplicity, we consider a single type of faults, i.e., $E_f=\lbrace f\rbrace$.

\begin{definition}\label{diagnosability}\cite{1995Sampath}
	Given a system $G$ and its generated language $L(G)$, $G$ is diagnosable with respect to (w.r.t.) the natural projection $P$ and the set of fault events $E_f$ if for all $f\in E_f$, and all $\lambda=\lambda_1f\in L(G)$, where $\lambda_1\in E^*$, there exists a finite number $n\in \mathbb{N}$ such that for all $\lambda_2\in E^*$ such that $\lambda\lambda_2\in L(G)$ and $|\lambda_2|\geq n$,  we have that for all $\lambda'\in P^{-1}(P(\lambda\lambda_2))\cap L(G)$, $f\in\lambda'$ holds (where $f\in \lambda'$ means that event $f$ appears in $\lambda'$).	
\end{definition}


\begin{definition}\label{opacity}\cite{2007Chris}
	Given a system $G$ and its generated language $L(G)$, $G$ is current-state opaque w.r.t. the natural projection $P$ and a set of secret states $S\subseteq X$ if for all $\lambda\in L(G)$ such that $f(x_0,\lambda)\cap S\neq\emptyset$, there exists $\lambda'\in L(G)$ such that $P(\lambda)=P(\lambda')$ and $f(x_0,\lambda')\cap(X-S)\neq\emptyset$.	
\end{definition}


\section{Event concealment Formulation}

In this work, we are interested in a scenario where a subset of events is deemed as the secret of a given system.
More specifically, we consider a system modeled as an NFA $G = (X, E, f, x_0)$ with event set $E=E_o\dot{\cup} E_{uo}$,  where $E_o$ and $E_{uo}$ are the sets of observable events and unobservable events, respectively.
We assume, without loss of generality, that the set of secret events $E_S$ is a subset of $E_{uo}$ (i.e., $E_S\subseteq E_{uo}$) since observable secret events will be trivially revealed. 
For the sake of simplicity, we consider a single type of secrets, i.e., $E_S=\lbrace s\rbrace$.
We make the following assumptions (similar to \cite{1995Sampath}) for the NFA $G$:\\
(1) The language of $G$ is live, which means that there is at least one event defined at each state $x\in X$.\\
(2) There exists no cycle of unobservable events in $G$.

\begin{definition}\label{defConcealable}
	Given a system $G$ and its generated language $L(G)$, the occurrence of secret event $s\in E_S$ is said to be concealable in $L(G)$ w.r.t. the natural projection $P$ if for each $\lambda=\lambda_1s\lambda_2\in L(G)$, where $\lambda_1,\lambda_2\in E^*$, there exists $\lambda'\in P^{-1}(P(\lambda))\cap L(G)$ such that $s\notin \lambda'$ holds.
	The system $G$ is concealable w.r.t. $P$ and $E_S$ if for all $s\in E_S$, $s$ is concealable.
\end{definition}

Suppose that a secret event $s\in E_S$ occurs after a sequence $\lambda_1$. Event $s$ is said to be concealable if after any indefinitely extendible continuation $\lambda_2\in E^*$, there exists a sequence $\lambda'$ that generates the same observations as $\lambda=\lambda_1s\lambda_2$ but contains no secret event $s$; this implies that the observer cannot be certain that the secret event $s$ has occurred.

\begin{definition}\label{defunConcealable}
	Given a system $G$ and its generated language $L(G)$, the occurrence of secret event $s\in E_S$ is said to be unconcealable in $L(G)$ w.r.t. the natural projection $P$ if there exists $\lambda_s=\lambda_1s\lambda_2\in L(G)$, where $\lambda_1,\lambda_2\in E^*$, and for all $\lambda'\in P^{-1}(P(\lambda_s))\cap L(G)$, $s\in \lambda'$ holds. 	Such a sequence $\lambda_s$ is called an $s$-revealing sequence.
	The system $G$ is unconcealable w.r.t. $P$ and $E_S$ if there exists at least one unconcealable secret event $s\in E_S$.
\end{definition}


\subsection{Comparison between Diagnosability and Concealability}

From Definition \ref{diagnosability}, a fault is said to be diagnosable if it can be detected within a finite number of observable events after its occurrence. In order to detect the fault accurately, one should ensure that the fault can be detected for {\em any} (long enough) execution after its occurrence. 
When considering a secret event, once it gets revealed under {\em some} execution after its occurrence, we say that the privacy of the system has been compromised (since there is a possibility for the secrecy of the event to be compromised).
Thus, there is clearly an inverse relationship between event diagnosis and event concealment. 
However, the two notions are not exactly the opposite of each other as clarified by the following proposition.
For the sake of comparison, we assume that the set of secret events is exactly the same as the set of fault events, i.e., the set of particular events is $E_p=E_S=E_f$ (again, we focus on a single type of secret/fault events).

\begin{proposition}\label{pro1}
	Given the system $G=(X,E,f,x_0)$ and its generated language $L(G)$, $G$ is not diagnosable w.r.t. the natural projection $P$, and a set of particular events $E_{p}=E_S=E_f\subseteq E_{uo}$ if $G$ is concealable w.r.t. $P$ and $E_p$.
\end{proposition}

The proof of Proposition \ref{pro1} is omitted since it follows directly from Definitions \ref{diagnosability} and \ref{defConcealable}.

\subsection{Comparison between Opacity and Concealability}

From Definition \ref{opacity}, current-state opacity requires that the intruder can never be certain, based on its observations, that the current state of the system is secret.
In this regard, one can convert concealability to current-state opacity in a manner similar to the reduction of fault diagnosis to a state isolation problem \cite{2020Chrisbook}.

Specifically, given a system $G=(X,E,f,x_0)$ and a set of secret events $E_S$, a modified automaton can be constructed as $G_S=(X\dot{\cup} X_S,E,f_S,x_0)$, where $X_S$ is a copy of the set of original states with label $S$ (i.e., $X_S=\lbrace1_S,2_S,\dots\rbrace$ for $X=\lbrace 1,2,\dots\rbrace$), and the transition function $f_S$ is defined as
\begin{enumerate}
	\item $f_S(x,e)=f(x,e)$ for $x\in X$, $e\in E-E_S$
	\item $f_S(x,e)=Q_S'$ if $Q'=f(x,e)$, for $x\in X$, $e\in E_S$  (here $Q'$ is a subset of states of the original system $G$, i.e., $Q'\subseteq X$, and $Q_S'$ is a copy of $Q'$ whose states are associated with label $S$, i.e., $Q_S'=\lbrace x_S'\in X_S|x'\in Q'\rbrace\subseteq X_S$). 
	\item $f_S(x_S,e)=Q_S'$ if $Q'=f(x,e)$, for $x_S\in X_S$, $e\in E$.
\end{enumerate}


Note that $G_S$ and $G$ are non-deterministic.
According to the construction of $G_S$, following a sequence of observations, one can isolate the current state of the system to be within the state set $X_S$ if and only if at least one secret event has occurred.
In other words, concealability can be converted to a current-state opacity problem.

\begin{proposition}\label{pro2}
	Given a system $G=(X,E,f,x_0)$ and a modified automaton $G_S=(X\dot{\cup} X_S,E,f_S,x_0)$, $G_S$ is current-state opaque w.r.t. the natural projection $P$ and the set of secret states $X_S$, if and only if $G$ is concealable w.r.t. $P$ and the set of secret events $E_S$.
\end{proposition}

The proof of Proposition \ref{pro2} is omitted since it follows directly from Definitions \ref{opacity} and \ref{defConcealable}, and Lemma 7.1 in \cite{2020Chrisbook}.

\subsection{Problem Formulation}
In the remainder of this paper, we consider a system modeled as an NFA $G = (X, E, f, x_0)$ with a singleton set of secret events $E_S=\lbrace s\rbrace\subseteq E_{uo}$, where $E=E_o\dot{\cup} E_{uo}$ (with $E_o$ being the set of observable events and $E_{uo}$ being the set of unobservable events).
We are interested in hiding from an external observer (a curious eavesdropper) confidential information of the system that is represented as the occurrences of events from $E_S$, which is called the secret event set.
Accordingly, the privacy of the system is the concealment of the occurrences of the secret events.
Although secret events are unobservable, their occurrences may be revealed after a finite number of subsequent observations.
The problem of event concealment is stated as follows.

\begin{problem}
	Given an NFA $G=(X,E,f,x_0)$ with $E=E_o\dot{\cup} E_{uo}$, determine whether or not $G$ is concealable  w.r.t. the natural projection $P$ and the set of secret events $E_{S}\subseteq E_{uo}$.
\end{problem}

\section{Verification of Concealability}
This section presents a necessary and sufficient condition for concealability based on a diagnoser construction, where the diagnoser was originally proposed for diagnosability verification \cite{1995Sampath}.
The idea is that a secret event is concealable if and only if for each arbitrarily long sequence in which the secret event occurs at least once, we can find at least one other arbitrarily long sequence in which the secret event does not appear and the two sequences have the same observable projection.
To proceed, several definitions are provided below.

We use label $l\in\lbrace N,S\rbrace$ to capture the diagnostic information at each state of the system, where label $S$ is associated with a state that can be reached via a sequence of events that contains the secret event $s\in E_S$, whereas label $N$ is associated with a state that can be reached via a sequence of events that does not contain the secret event.
The diagnoser of the given system $G=(X,E,f,x_0)$ is derived from \cite{1995Sampath} as $G_d = (X_{d}, E_o, f_d, x_{d_0})$, where $X_d\subseteq 2^{X\times \lbrace N, S\rbrace}$ is the set of states, $x_{d_0}=\lbrace (x,S)|\exists u=u'su''\in E_{uo}^*,s\in E_S, x\in f(x_0,u)\rbrace\dot{\cup}\lbrace (x,N)|\exists v\in(E_{uo}-E_S)^*,x\in f(x_0,v)\rbrace$
is the initial state, and $f_d$ is the transition function constructed in a manner similar to the transition function of an observer \cite{1990ozeveren}, which includes attaching secret labels $l$ to the states and propagating these labels from state to state. Note that the resulting state space $X_d$ is a subset of $2^{X\times \lbrace N, S\rbrace}$ composed of the states of the diagnoser that are reachable from $x_{d_0}$ under~$f_d$. 

A state $x_d$ in $X_d$ is of the form $x_d=\lbrace (x_1,l_1),(x_2,l_2),...,(x_n,l_n)\rbrace$, where $x_i\in X$ and $l_i\in \lbrace N,\allowbreak S\rbrace$ for all $i\in\lbrace 1,2,\dots,n\rbrace$.
A state $x_d$ is called normal if $l_i = N$ for all $(x_i,l_i)\in x_d$; $x_d$ is called secret if $l_i = S$ for all $(x_i,l_i)\in x_d$; and $x_d$ is called uncertain if there exist $(x_i,l_i),(x_j,l_j)\in x_d$ such that $l_i = N$ and $l_j = S$, where $i, j\in\lbrace 1,2,\dots,n\rbrace$.
The state space $X_d$ can be accordingly partitioned into three disjoint sets $X_d=X_{N_d}\dot{\cup} X_{S_d}\dot{\cup} X_{NS_d}$, where $X_{N_d}$ is the set of normal states, $X_{S_d}$ is the set of secret states, and $X_{NS_d}$ is the set of uncertain states.

The revealing of a secret is captured by a state labeled with $S$ in the diagnoser.
Since revealing is irreversible, subsequent states in the diagnoser (that can be reached from this state) will also be labeled with $S$.  
Thus, to verify concealability, we argue that a given system is concealable as long as no secret state is found in its diagnoser.
%

\begin{lemma}\label{lemmaSstate}
	Given a system $G=(X,E,f,x_0)$ and its diagnoser $G_d=(X_d,E_o,f_d,x_{d_0})$, there exists a secret state in $G_d$ if and only if there exists an $s$-revealing sequence $\lambda_s=\lambda_1s\lambda_2\in L(G)$, where $s\in E_S$.
\end{lemma}

\textit{Proof:}
Necessity: The existence of a secret state in $G_d$ implies that the occurrence of the secret event is revealed following a sequence of observations, which leads to that state. Furthermore, by the diagnoser construction, if there exists a secret state in $G_d$, there necessarily exists in the system an $s$-revealing sequence that generates a sequence of observations such that the secret state is reached.	

Sufficiency: If there exists an $s$-revealing sequence $\lambda_s$, the occurrence of the secret event will be revealed.
Once the occurrence of the secret event is detected, the current diagnostic state will necessarily be labeled as $S$.
Thus, there exists at least one secret state in $G_d$. 
\hfill
$\diamond$

%
%

\begin{proposition}
	Given a system $G=(X,E,f,x_0)$ and its diagnoser $G_d=(X_d,E_o,f_d,x_{d_0})$, $G$ is concealable w.r.t. $P$ and $E_S\subseteq E_{uo}$ if and only if there exists no secret state in~$G_d$.
\end{proposition}

\textit{Proof:}
Necessity:
		We first prove that if $G$ is concealable, then there exists no secret state in $G_d$. 
		By following Definition~\ref{defConcealable}, $G$ is concealable if for all $s\in E_S$, and all $\lambda=\lambda_1s\lambda_2\in L(G)$, there exists $\lambda'\in P^{-1}(P(\lambda))\cap L(G)$ such that $s\notin \lambda'$ holds. This means that there exists no $s$-revealing sequence, which implies the absence of a secret state in $G_d$ due to Lemma \ref{lemmaSstate}.

	Sufficiency:
       We now prove that if there exists no secret state in $G_d$, then $G$ is concealable.
       By contrapositive, assume that $G$ is unconcealable.
       Then, there exists a corresponding $s$-revealing sequence $\lambda_s=\lambda_1s\lambda_2$ with $s\in E_S$ according to Definition \ref{defunConcealable}.
       In conclusion, there exists a secret state in $G_d$ by following Lemma~\ref{lemmaSstate}.
\hfill
$\diamond$

\begin{figure}[htbp]
	\centering  
	\subfigure[System $G$.]{
		\centering
		\includegraphics[width=0.28\textwidth]{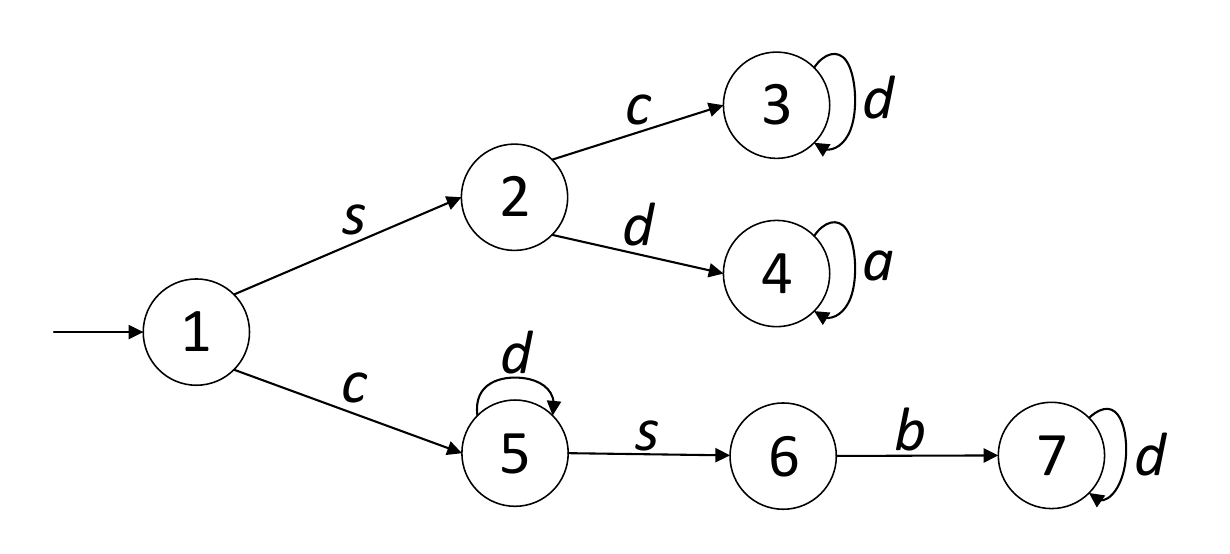}
		\label{G}
	}
	\subfigure[Diagnoser $G_d$.]{
		\centering
		\includegraphics[width=0.15\textwidth]{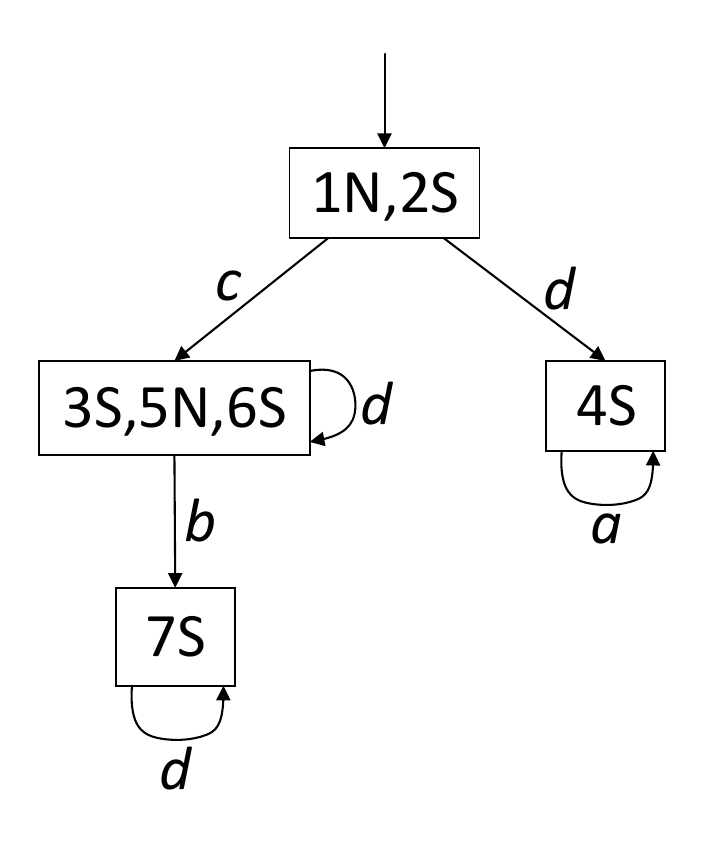}
		\label{Gd}
	}
	\caption{Verification of concealability.}
	\label{Verificationofconcealability}
\end{figure}

\begin{eg}
	Consider the system $G=(X,E,f,x_0)$ as shown in Fig. \ref{G}, where $X=\lbrace1,2,3,4,5,6,7\rbrace$, $x_0=\lbrace1\rbrace$, and $E=E_o\dot{\cup} E_{uo}$ with $E_o=\lbrace a,b,c,d\rbrace$ and $E_{uo}=\lbrace s\rbrace$. 
	Assume that the set of secret events $E_S$ equals $E_{uo}$, i.e., $E_S=E_{uo}=\lbrace s\rbrace$.
	The diagnoser $G_d$ (shown in Fig. \ref{Gd}) is driven by the observable events in $G$.
	One can notice that there are two secret states, namely, $\lbrace 4S\rbrace$ and $\lbrace 7S\rbrace$ in $G_d$.
	Hence, system $G$ is unconcealable since an eavesdropper is able to detect the occurrence of the secret event if states $\lbrace 4S\rbrace$ and $\lbrace 7S\rbrace$ are reached.
	\hfill$\diamond$
\end{eg}

\begin{remark}(A verifier cannot be used to verify concealability) 
	A verifier as presented in \cite{2001Kumar}, \cite{2002Lafortune} is a construction of polynomial complexity that can be used to verify diagnosability.
	However, a verifier cannot be used to verify concealability because it tracks the behavior of system $G$ in a pairwise manner and does not allow us to be conclusive about concealability.
	To verify concealability, one needs to check whether the occurrences of secret events have been detected in \emph{all} diagnostic states of a diagnoser construction.
	Intuitively, the presence of a secret state in a diagnoser leads to the conclusion that the occurrences of secret events have been detected; that is, the system $G$ is unconcealable.
	However, the presence of a secret state in a verifier does not necessarily imply the violation of concealability.
	Thus, a verifier cannot be used to verify the concealability of a system.
	Nevertheless, as we argue later in the paper, a verifier-based approach can still be used to obtain necessary and/or sufficient conditions for $C$-enforceability.
\end{remark}

\section{Concealability Enforcement Formulation}

To deal with the enforcement of concealability, we assume that the considered system is unconcealable in the remainder of this paper.
In this section, we first introduce a defensive function to manipulate actual observations generated by the system in order to enforce concealability.
Then, the notion of $C$-enforceability is introduced, which characterizes the ability of a defensive function to manipulate the observations of output events such that the occurrences of secret events can be concealed from the eavesdropper regardless of system activity.

\subsection{Defensive Function}
A defensive function is an interface of a system to obfuscate the observations seen at the eavesdropper by appropriately modifying the original observable events generated by the system (via deletions, insertions, or replacements).
Note that we assume that the eavesdropper is not aware of the existence of the defensive function in this paper.
We denote the set of replaced events by $E_R\subseteq E_o\dot{\cup} \lbrace\varepsilon\rbrace$, whereas the set of events that can be inserted is denoted as $E_I\subseteq E_o$.
A defensive function, denoted by $D:E_o^*\rightarrow 2^{E_o'^*}$, is defined recursively as $D(\lambda t)=D(\lambda)D(t)$ for $\lambda\in E_o^*$ and $t\in E_o$, with $D(t)\in E_o'$, where $E_o'=\lbrace t/o\mid o\in E_R\dot{\cup} \lbrace et\mid e\in E_I\rbrace\rbrace$.
Notation ``$/$'' is used to represent defensive operations for each defensive action.
That is, the defensive function can capture the set of possible sequences of defensive actions for a given observation sequence generated by the system.
A defensive projection\footnote{For example, consider a set of observable events $E_o=\lbrace a,b\rbrace$ and assume that the set of replaced events is $E_R=\lbrace a,\varepsilon\rbrace$ and the set of inserted events is $E_I=\lbrace b\rbrace$. The defensive action for event $a$ satisfies $D(a)\in\lbrace a/a, a/\varepsilon,ba\rbrace$. The eavesdropper may receive any observation from $P_D(D(a))=\lbrace a,\varepsilon,ba\rbrace$ after event $a$ occurs, depending on the choice we make for $D(a)$.} $P_D: E_o'^*\rightarrow E_o^*$ strips the notation $/$ out from $D$ and is defined recursively as $(i)$ $P_D(\sigma)=o$ where $\sigma=t/o$ and $P_D(\sigma)=et$ where $\sigma=t/et$ for $t\in E_o$, $o\in E_R$, $e\in E_I$, and $(ii)$ $P_D(\omega\sigma)=P_D(\omega)P_D(\sigma)$ for $\omega\in E_o'^*$ and $\sigma\in E_o'$ (we take $P_D(\varepsilon)=\varepsilon$).

From a practical viewpoint, there may exist some constraints in the use of the defensive function (e.g., the intruder may have direct access to a sensor that generates a certain label, therefore that label can never be replaced or inserted).
To capture such constraints, we consider the following scenario in the remainder of the paper:
each event $t$ in $E_o$ that is generated by the system can be deleted or replaced with \emph{some} of the other events in $E_o$ but not necessarily all events ($E_R\subseteq E_o\dot{\cup}\lbrace\varepsilon\rbrace$); furthermore, only some events in $E_o$ (but not necessarily all events) can be inserted before the output of event $t$ ($E_I\subseteq E_o$).

\subsection{$C$-Enforceability}

\begin{definition}\label{enforceable}
	Given a system $G$ and its generated language $L(G)$, a defensive function $D$ is said to be $C$-enforcing w.r.t. $P$ and $E_{S}$ if the following conditions hold:
	\begin{enumerate}
		\item for all $\lambda\in L(G)$, $D(P(\lambda))$ is defined;
		\item for each $s$-revealing sequence $\lambda_s=\lambda_1s\lambda_2\in L(G)$,  where $s\in E_S$, $\lambda_1, \lambda_2\in E^*$, there exists at least one sequence $u=t_1t_2\dots t_n\in L(G)$ with $t_i \notin E_S$ for $i\in\lbrace1,2,\dots,n\rbrace$, such that $P(u)\in P_D(D(P(\lambda_s)))$; furthermore, for each continuation $\lambda'\in E^*$ such that $\lambda_s\lambda'\in L(G)$, we can find $u'=t_1't_2'\dots t'_{n'}\in L(G)$ with $t_i'\notin E_S$, such that $P(u')\in P_D(D(P(\lambda_s\lambda')))$ and $P(u)\in \overline{P(u')}$.
	\end{enumerate}
\end{definition}

The first condition requires that,
as the interface of the system, the defensive function should be able to react to every observable sequence of events that can be  generated by the system, such that defensive actions (deletions, insertions, or replacements) can be utilized.
In other words, in order to retain the privacy of a system, for each $s$-revealing sequence $\lambda_s$, the corresponding defensive actions should include a feasible one that does not disclose the occurrences of the secret events as described in the second condition.
Furthermore, for the sake of maintaining accuracy and completeness of defensive actions, the last part of the second condition guarantees that, for each subsequent activity, defensive actions can retain the concealment of the occurrences of the secret events indefinitely.
If both conditions are satisfied, the concealability of the system can be enforced by a $C$-enforcing defensive function.

\begin{definition}
	Given a system $G$ and its generated language $L(G)$, $G$ is said to be $C$-enforceable w.r.t. $P$ and $E_{S}$ if there exists a $C$-enforcing defensive function $D$.
\end{definition}

%

\subsection{Problem Statement}
Consider a system modeled as an NFA $G=(X,E,f,x_0)$, where $E=E_o\dot{\cup} E_{uo}$ with $E_o$ and $E_{uo}$ being the sets of observable and unobservable events, respectively, and let $E_S\subseteq E_{uo}$ be the set of secret events. 
We consider the case that the system is unconcealable. 
Then, under some system behavior, the eavesdropper can infer the occurrences of the secret events by utilizing its knowledge of the system model and by analyzing the observable behavior of the system.
The defensive function proposed in this section can alter observable output events of the system $G$ by deletions, insertions, or replacements. 
The problem of enforcing concealability of the system aims to determine whether the defensive function is $C$-enforcing, i.e., given constraints in terms of how the defensive function can delete, insert or replace events, determine whether it is possible to conceal the occurrences of the secret events.
This problem can be formalized as follows.

\begin{problem}
	Given an NFA $G=(X,E,f,x_0)$ with $E=E_o\dot{\cup} E_{uo}$, a set of secret events $E_{S}\subseteq E_{uo}$, and a defensive function $D$, determine whether the defensive function $D$ is $C$-enforcing.
\end{problem}

\section{Verification of $C$-enforceability}

In this section, a verifier and a defensive verifier are constructed to respectively capture system behavior and all feasible defensive actions following system activity. 
Then, an $E$-verifier is built by a special synchronization mechanism between a verifier and a defensive verifier to verify $C$-enforceability for a defensive function that has to operate under deletion, insertion, and replacement constraints.
Specifically, the $E$-verifier can be used to obtain, with polynomial complexity, one necessary and one sufficient condition for $C$-enforceability; in case that the sufficient condition is satisfied, the trimmed version of the $E$-verifier leads to a strategy to enforce concealability, also with polynomial complexity.
These developments should be contrasted against constructions with exponential complexity \cite{2014Wu} (the latter, however, provide a necessary and sufficient condition).

\subsection{Verifier and Defensive Verifier}

In order to detect/identify secret events with reduced complexity, the verifier is constructed next by following an approach similar to the one proposed in \cite{2001Kumar}, described via the two-step procedure below.
\begin{enumerate}
	\item First, we construct an NFA $G_o=(X_o,E_o,f_o,x_0^o)$ from $G$ with $L(G_o)=P(L(G))$, where $X_o=\lbrace (x_i,l_i)|x_i\in X, l_i\in \lbrace N,S\rbrace\rbrace$ is the finite set of states, $E_o$ is the set of observable events, $x_0^o=(x_0,N)$ is the initial state, and $f_o: X_o\times E_o\rightarrow 2^{X_o}$ is the transition function for $x\in X_o$, $e\in E_o$, defined as follows.
	
	For $l=N$, 
	\begin{equation*}
	f_o((x,l),e)=\left\{
	\begin{aligned}
	&\lbrace(x',N)\mid\exists t\in (E_{uo}-E_S)^*,\\
	&\qquad\qquad P(t)=e,x'\in f(x,t)\rbrace,  \\ 
	&\lbrace(x',S)\mid\exists u=u'su''\in E_{uo}^*,\\
	&\qquad\qquad P(u)=e, x'\in f(x,u)\rbrace, \\ 
	\end{aligned}
	\right.
	\end{equation*}
	For $l=S$,
	\begin{equation*}
	f_o((x,l),e)=\left\{
	\begin{aligned}
	&\lbrace(x',S)|\exists v\in E_{uo}^*,\\
	&\qquad\qquad P(v)=e, x'\in f(x,v)\rbrace.   \\ 
	\end{aligned}
	\right.
	\end{equation*}
	\item Then, a verifier, captured by $G_{v}$, can be obtained by composing $G_o$ with itself. It is defined as $G_{v}=(X_{v},E_o,f_{v}, x_{v_0})$, where $X_{v}=X_o\times X_o$ is the set of states, $E_o$ is the set of observable events, $f_{v}: X_v\times E_o\rightarrow 2^{X_v}$ is the transition function defined by $f_v((x_i,l_i),(x_j,l_j),e)=f_o((x_i,l_i),e)\times f_o((x_j,l_j),e)$ for all $e\in E_o$ if and only if both $f_o((x_i,l_i),e)$ and $f_o((x_j,l_j),e)$ are defined, and $x_{v_0}=x_0^o\times x_0^o=\lbrace((x_0,N),(x_0,N))\rbrace\in X_v$ is the initial state. 
\end{enumerate}

Each state $x_v\in X_v$ can be seen as a pair $\lbrace(x_i,l_i),(x_j,l_j)\rbrace$, where $x_i,x_j\in X$ and $l_i,l_j\in\lbrace N,S\rbrace$. 
We can partition $X_v$ into three disjoint sets in a similar manner as done for the diagnoser.
State $x_v$ is called normal if $l_i=l_j=N$; it is called secret if $l_i=l_j=S$; and it is called uncertain if $l_i\neq l_j$. 
Thus, the state space $X_v$ can be partitioned into three disjoint sets $X_v=X_{N_v}\dot{\cup} X_{S_v}\dot{\cup} X_{NS_v}$, where $X_{N_v}$ is the set of normal states, $X_{S_v}$ is that of secret states, and $X_{NS_v}$ is that of uncertain states.


For the given system $G$, we denote the set of possible defensive actions outputted via the defensive function under deletion, insertion, and replacement constraints by $E_D$, which is defined as $E_D=\bigcup_{t\in E_o}D(t)$.
A defensive verifier ${G_D}$ is defined below to capture all feasible defensive actions.
%

\begin{definition}\label{G_{ad}}
	Given an NFA $G=(X,E,f,x_0)$ and its verifier $G_v=(X_v,E_o,f_v,x_{v_0})$, a defensive verifier is an NFA $G_D=(X_D,E_D,f_D,x_{D_0})$, where $X_D=X_v-X_{S_v}$ is the state space, $x_{D_0}=x_{v_0}$ is the initial state, $E_D$ is the set of defensive actions with constraints, and $f_D$ is the transition function that implements defensive actions for $t/o\in E_D$, and $x_D\in X_D$: 
	\begin{equation*}
	f_D(x_D,t/o)=\left\{
	\begin{array}{rll}
	&f_v(x_D,\varepsilon) \\
	& \quad\text{if}\quad P_D(t/o)=\varepsilon \quad\text{(deletion)},\\ 
	&f_v(x_D,t')\cap X_D \\
	& \quad \text{if}\quad P_D(t/o)=t' \quad\text{(replacement)},\\
	&(\bigcup_{x_D'\in f_v(x_D,t')}f_v(x_D',t))\cap X_D \\
	& \quad \text{if}\quad P_D(t/o)=t't \quad\text{(insertion)}.\\  
	\end{array}
	\right.
	\end{equation*}
\end{definition}

Informally, the defensive verifier $G_D$ can be constructed by cloning $G_v$ and pruning away all secret states and all original events in $E_o$ while adding defensive events, and then taking the accessible part of the resulting automaton.

\begin{figure}[htbp]
	\centering  
	\subfigure[Verifier $G_v$.]{
		\centering
		\includegraphics[width=0.45\textwidth]{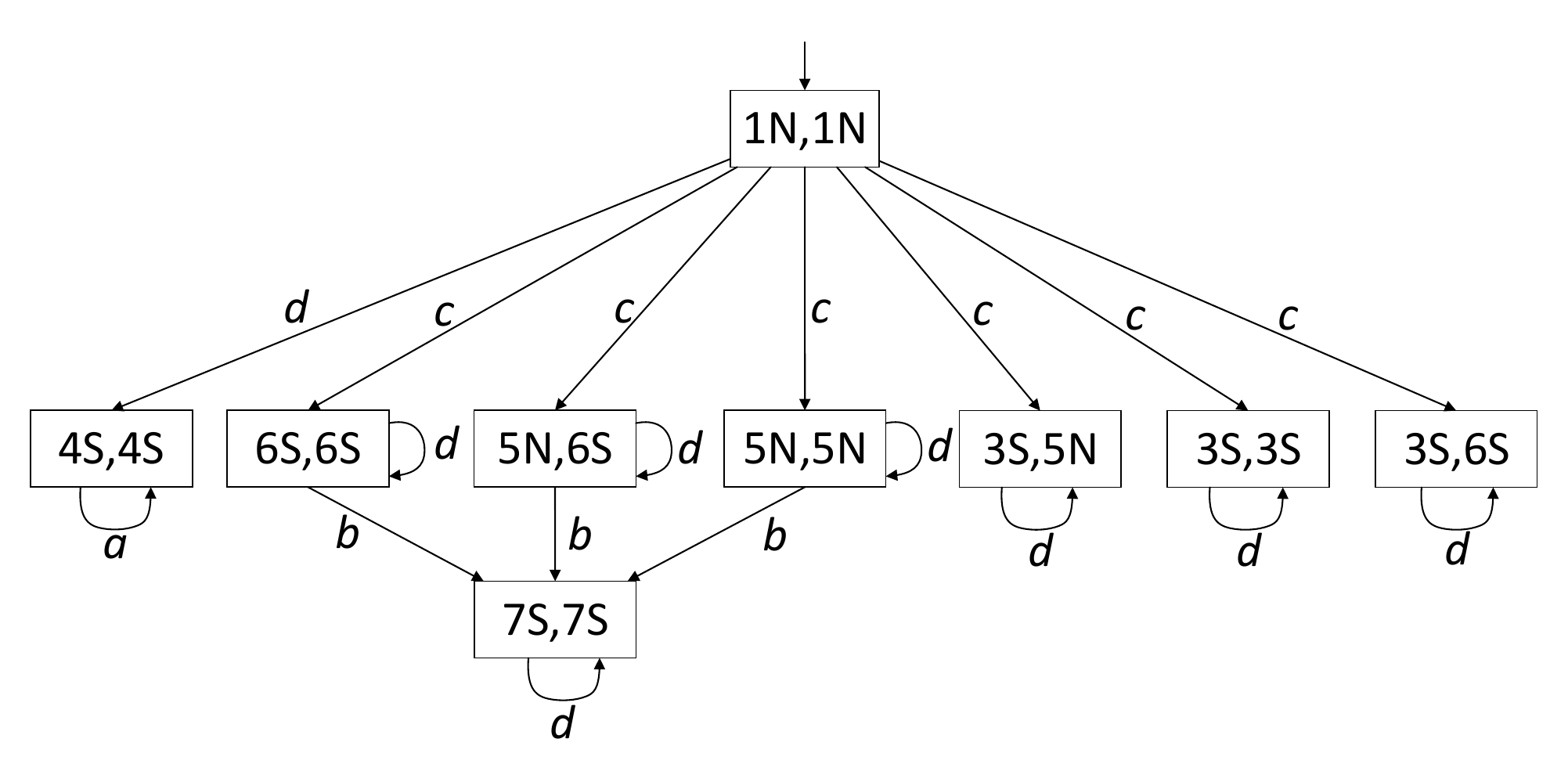}
		\label{Gv}
	}

	\subfigure[Defensive verifier $G_D$.]{
		\centering
		\includegraphics[width=0.25\textwidth]{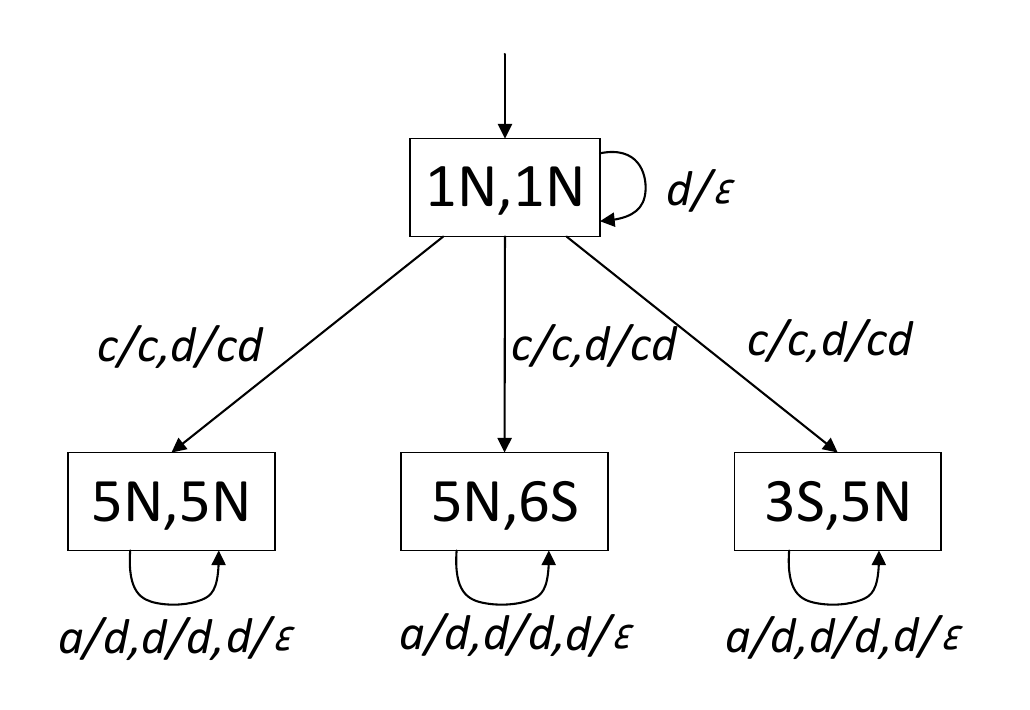}
		\label{GD}
	}
	\caption{Verifier and defensive verifier constructions.}
	\label{GvGD}
\end{figure}

%

The verifier and the defensive verifier of the system in Example 1 are shown in Fig. \ref{GvGD} (note that, following common practice, only one of states $(x_i,x_j)$ and $(x_j,x_i)$ is retained in the verifier for brevity, i.e., states $(3S,5N), (3S,6S)$, and $(5N,6S)$).
For the set of observable events $E_o=\lbrace a,b,c,d\rbrace$, we impose constraints on corresponding defensive actions as follows: $D(a)\in\lbrace a/a,a/d\rbrace$, $D(b)\in\lbrace b/b\rbrace$, $D(c)\in\lbrace c/c\rbrace$, and $D(d)\in\lbrace d/d,d/cd,d/\varepsilon\rbrace$.

Since the set of states of $G_D$ is identical to (or a subset of) the set of states of $G_v$, we can also classify the states of $G_D$ in two disjoint sets: the set of normal states $X_{N_v}$, and the set of uncertain states $X_{NS_v}$ (recall that the set of secret states has been pruned away). 

\subsection{Verification of $C$-Enforceability under Defensive Functions with Constraints}
A $C$-enforcing defensive function should ensure that all possible defensive actions keep the eavesdropper confused regardless of system activity.
Thus, when the defensive function is subject to constraints, we propose a new construction by composing a verifier and a defensive verifier of a given system to capture all feasible defensive actions following system activity.

\begin{definition}\label{RdefinitionKverifier}
	Given an NFA $G=(X,E,f,x_0)$ with its verifier $G_v=(X_v,E_o,f_v,x_{v_0})$ and its defensive verifier $G_{D}=(X_D,E_D,f_D,x_{D_0})$, the $E$-verifier is an NFA that captures the composition of $G_v$ and $G_D$, denoted by $V_E=(X_E, E_D, f_E, x_{E_0})$, satisfying the following:
	\begin{enumerate} 
		\item 	$X_E=\lbrace (x_v,x_D)\mid  x_v\in X_v, x_D\in X_D\rbrace$ is the state space;
		\item  $x_{E_0}= (x_{v_0},x_{D_0})$ is the initial state;
		\item $E_D$ is the set of defensive actions;
		\item $f_E$ is the transition function defined as follows: 
		$$f_E((x_v,x_D),t/o)=f_v(x_v,t)\times f_D(x_D,t/o)$$  $\noindent$ where $t\in E_o$, $t/o\in D(t)$.
	\end{enumerate}	
\end{definition}

The main idea underlying Definition $\ref{RdefinitionKverifier}$ is to compute a parallel composition of $G_v$ and $G_D$, where the event synchronization is performed w.r.t. $E_o$ from $G_v$ and $E_D$ from~$G_D$.

\begin{definition}
	Given an $E$-verifier $V_E=($$X_E$, $E_D,$ $f_E,$ $x_{E_0})$, a defensive action $t/o\in E_D$ is said to be feasible at state $(x_v,x_D)\in X_E$ if $f_E((x_v,x_D),t/o)\ne\emptyset$.
\end{definition}

\begin{lemma}
	Given an NFA $G=(X,E,f,x_0)$ with its~verifier $G_v=(X_v,E_o,f_v,x_{v_0})$ and defensive verifier $G_{D}=(X_D, E_D,f_D,x_{D_0})$, the $E$-verifier $V_E$ enumerates feasible defensive actions following system activity.
\end{lemma}

The proof of the above lemma follows directly from Definition $\ref{RdefinitionKverifier}$ and is thus omitted.

Note, for state $(x_v,x_{D})$ in the $E$-verifier $V_E$, the first component represents a pair of possible actual states\footnote{Recall that the verifier $G_v$ is able to capture observable behavior of the system by tracking pairs of states in $G_o$ (that the system might reach).} that the system can be in at current time, whereas the second one represents a pair of possible fake states perceived by the eavesdropper after defensive actions. 
Since $G_v$ and $V_E$ are nondeterministic, multiple states in $G_v$ and $V_E$ might be reachable at current time following a sequence of observations.
We denote the set of co-relative states in the verifier $G_v$ for state $x_o\in X_o$ by $X_{v_{x_o}}=\lbrace (x_o,x_o')\in X_v|x_o'\in X_o\rbrace\cup\lbrace (x_o',x_o)\in X_v|x_o'\in X_o\rbrace $.
That is, the co-relative states of $x_o$ form the set $X_{v_{x_o}}$ composed of all states of the verifier having $x_o$ as one of its components.
We denote the set of co-relative states\footnote{The set $X_{v_{x_o}}$ ($X_{E_{x_o}}$) contains all possible actual states (pairs of actual and fake states) in the verifier ($E$-verifier) that might be simultaneously reachable following a sequence of observations that leads $G_o$ to state $x_o$.} in the $E$-verifier $V_E$ for state $x_o\in X_o$ by $X_{E_{x_o}}=\lbrace (x_v,x_D)\in X_E|x_v\in X_{v_{x_o}}\rbrace$.

\begin{definition}
	Consider an NFA $G=(X,E,f,x_0)$ with $G_o=(X_o, E_o,f_o,x_0^o)$ and $E$-verifier $V_E=($$X_E$, $E_D,$ $f_E,$ $x_{E_0})$.
		The set of states $X_{E_{x_o}}$ for state $x_o\in X_o$ is said to be legal if for all $t\in E_o$ such that $f_o(x_o,t)\neq\emptyset$ there exists $(x_v,x_D)\in X_{E_{x_o}}$ and defensive action $t/o\in E_D$ such that $f_E((x_v,x_D), t/o)$ is defined; $X_{E_{x_o}}$ is said to be illegal otherwise.
\end{definition}

For the defensive function to react to system activity, we need that for each (reachable) system state $x_o\in X_o$ and for each $t\in E_o$ that is feasible at $x_o$ ($f_o(x_o,t)\neq\emptyset$), we can find among all states in $X_{E_{x_o}}$ at least one defensive action $t/o$ that is allowed. 
If we cannot, the co-relative state set $X_{E_{x_o}}$ is illegal.
In other words, at state $x_o$ in $G_o$ (assumed reachable without loss of generality), we will have no way of maintaining secrecy following observable event $t$ that can be generated by the~system.

\begin{theorem}
		Consider an NFA $G=(X,E,f,x_0)$ with $G_o=(X_o,E_o,f_o,x_0^o)$, a defensive function $D$, and an $E$-verifier $V_E=(X_{E}, E_D, f_E, x_{E_0})$. If the defensive function $D$ is $C$-enforcing, then the set of co-relative states $X_{E_{x_o}}$ is legal for all (reachable) $x_o\in X_o$.
\end{theorem}

\textit{Proof:}
	By contrapositive, suppose that there exists a state $x_o\in X_o$ such that the set of its co-relative states $X_{E_{x_o}}$ is illegal. This means that 
	for a system action $t\in E_o$ that is defined at state $x_o$, all corresponding states $(x_v,x_D)\in X_{E_{x_o}}$ satisfy $f_E((x_v,x_D), t/o)=\emptyset$ for all $t/o$.
	In this case, there is no defensive action that can be taken in response to system action $t$, which violates the first condition of Definition 5, i.e., the defensive function $D$ is not $C$-enforcing.
	~\hfill
$\diamond$

\begin{algorithm}[htbp]
	\caption{Necessary condition for $C$-enforceability of a defensive function}
	\label{AlgorithmE}
	\LinesNumbered
	\KwIn{$G_o=(X_o,E_o,f_o,x_0^o)$ and $E$-verifier $V_E=(X_{E}, E_D, f_E,  x_{E_0})$, where $X_{E}=\lbrace (x_v,x_D)\mid  x_v\in X_v, x_D\in X_D\rbrace$.}
	\For{$x_o\in X_o$}{
		Let $X_{E_{x_o}}$ be the set of co-relative states for state $x_o$, i.e., $X_{E_{x_o}}=\lbrace (x_v,x_D)\in X_E|x_v\in X_{v_{x_o}}\rbrace$\;
		\For{$t\in E_o$ such that $f_o(x_o,t)\neq\emptyset$}{
			\If{$\forall$$(x_v,x_D)\in X_{E_{x_o}}$,$\forall t/o\in D(t)$, $f_E((x_v,x_D),t/o)$ \textcolor{blue}{$=$} $\emptyset$}{
				\Return Not $C$-enforcing.
			}
		}
	}
	\Return Possibly $C$-enforcing.
\end{algorithm}

Given an $E$-verifier, we can check the necessary condition for the defensive function to be $C$-enforcing by following Algorithm~1. 
However, it is possible that a defensive function may not be $C$-enforcing even though the necessary condition is satisfied.
To this end, a sufficient condition for $C$-enforceability is provided next.


Note that the $E$-verifier may generate some problematic states, i.e., for state $(x_v,x_D)\in X_E$, there may not exist a subsequent state that can be reached for some $t\in E_o$ that is defined at state $x_v$.
In other words, there is no defensive action that can be taken in response to system execution $t$ for some $t\in E_o$.
Moreover, the removal of such states may\footnote{For instance, consider a state $(x_v,x_D)\in X_E$ and an event $t\in E_o$ for which there are two defensive actions $t/o_1$ and $t/o_2$ such that $f_E((x_v,x_D),t/o_1)=\lbrace (x_{v_1},x_{D_1})\rbrace$ and $f_E((x_v,x_D),t/o_2)=\lbrace (x_{v_2},x_{D_2})\rbrace$.
If state $(x_{v_1},x_{D_1})$ is problematic, its removal will not make state $(x_v,x_D)$ problematic.} generate new problematic states.
Note that when the same defensive action evolves from some originating $E$-verifier state to multiple destination states (on account of the nondeterminacy of the $E$-verifier), if one of the multiple destination states is problematic, then only that state needs to be avoided.\footnote{Assume that for a given state $(x_v,x_D)\in X_E$ and an event $t\in E_o$, there is a defensive action $t/o\in D(t)$ such that $f_E((x_v,x_D),t/o)=\lbrace (x_{v_1},x_{D_1}), (x_{v_2},x_{D_2})\rbrace$, where state $(x_{v_1},x_{D_1})$ is problematic.
	In this case, the removal of state $(x_{v_1},x_{D_1})$ will not make state $(x_v,x_D)$ problematic: since state $(x_{v_2},x_{D_2})$ is retained, it can be reached from state $(x_v,x_D)$ via the defensive action $t/o$.}
In order to ensure that the $E$-verifier captures only feasible defensive actions following system activity, a formal procedure for constructing the reduced $E$-verifier, denoted by $RV_E$, by removing problematic states in an iterative way, is presented in Algorithm 2 below.

\begin{algorithm}[htbp]
	\caption{Construction of reduced $E$-verifier}
	\label{AlgorithmVE}
	\LinesNumbered
	\KwIn{$E$-verifier $V_E=(X_{E}, E_D, f_E,  x_{E_0})$, where $X_{E}=\lbrace (x_v,x_D)\mid  x_v\in X_v, x_D\in X_D\rbrace$.}
	\KwOut{Reduced $E$-verifier $RV_E=(X_{ER}, E_D, f_E,  x_{E_0})$.}
	$X_{P}=\emptyset$, where $X_{P}$ denotes the set of problematic states\;
	\For{$(x_v,x_D)\in X_E-X_P$}{
		\For{$t\in E_o$ such that $f_v(x_v,t)\neq\emptyset$}{
			\If{$\forall t/o\in D(t)$, $f_E((x_v,x_D),t/o)\cap (X_E-X_P)=\emptyset$}{
				$X_P:=X_P\dot{\cup} \lbrace(x_v,x_D)\rbrace$\;
			}
		}
	}	
	Go back to Step 2 until no more new problematic states are produced\;
	$X_{ER}=X_E-X_P$\;
	\Return $RV_E$.
\end{algorithm}


\begin{theorem}
		Consider an NFA $G=(X,E,f,x_0)$, a defensive function $D$, and a reduced $E$-verifier $RV_E=(X_{ER}, E_D, f_E, x_{E_0})$. If there exists $(x_v,x_D)\in X_{ER}$ for all $x_v\in X_v$, then the defensive function $D$ is $C$-enforcing.
\end{theorem}

\textit{Proof:}
	State $(x_v,x_D)\in X_{ER}$ contains two components, where the first satisfies $x_v\in X_v$ and the second satisfies $x_D\in X_D$ (i.e., $x_v$ and $x_D$ are states of the verifier $G_v$ and the defensive verifier $G_D$, respectively).
	If there exists $(x_v,x_D)\in X_{ER}$ for all $x_v\in X_v$, this means that each observable event generated by the verifier has at least one defensive action; hence the first condition of Definition~\ref{enforceable} holds.
	Moreover, each subsequently visited state in $x_{ER}$ has the same property (otherwise, it would have been pruned away).
	The case that $x_v\in X_{S_v}$ while $x_D\in X_{S_v}$ for each state $(x_v,x_D)\in X_{ER}$ is not possible since the secret states were pruned away in $G_D$ (in fact, $x_D\notin X_{S_v}$ regardless of $x_v$), i.e., the second condition of Definition \ref{enforceable} holds.
	Therefore, one concludes that the defensive function $D$ is $C$-enforcing.
\hfill
$\diamond$

Note that Theorem 2 provides a sufficient condition for $C$-enforceability. 
If the condition holds, the reduced $E$-verifier can be used to take defensive actions in response to system activity.
In fact, any action that is feasible at the present state —or one of the present states— of the reduced $E$-verifier can be taken since there will always be a feasible action in all of the future states, regardless of the event that occurs.
This also gives us a simple way to obtain a strategy for the defensive function so as to enforce concealability of the system; this strategy is captured by the reduced $E$-verifier in a non-deterministic manner.

\begin{figure}[htbp]
	\centering
	{
		\includegraphics[width=0.5\textwidth]{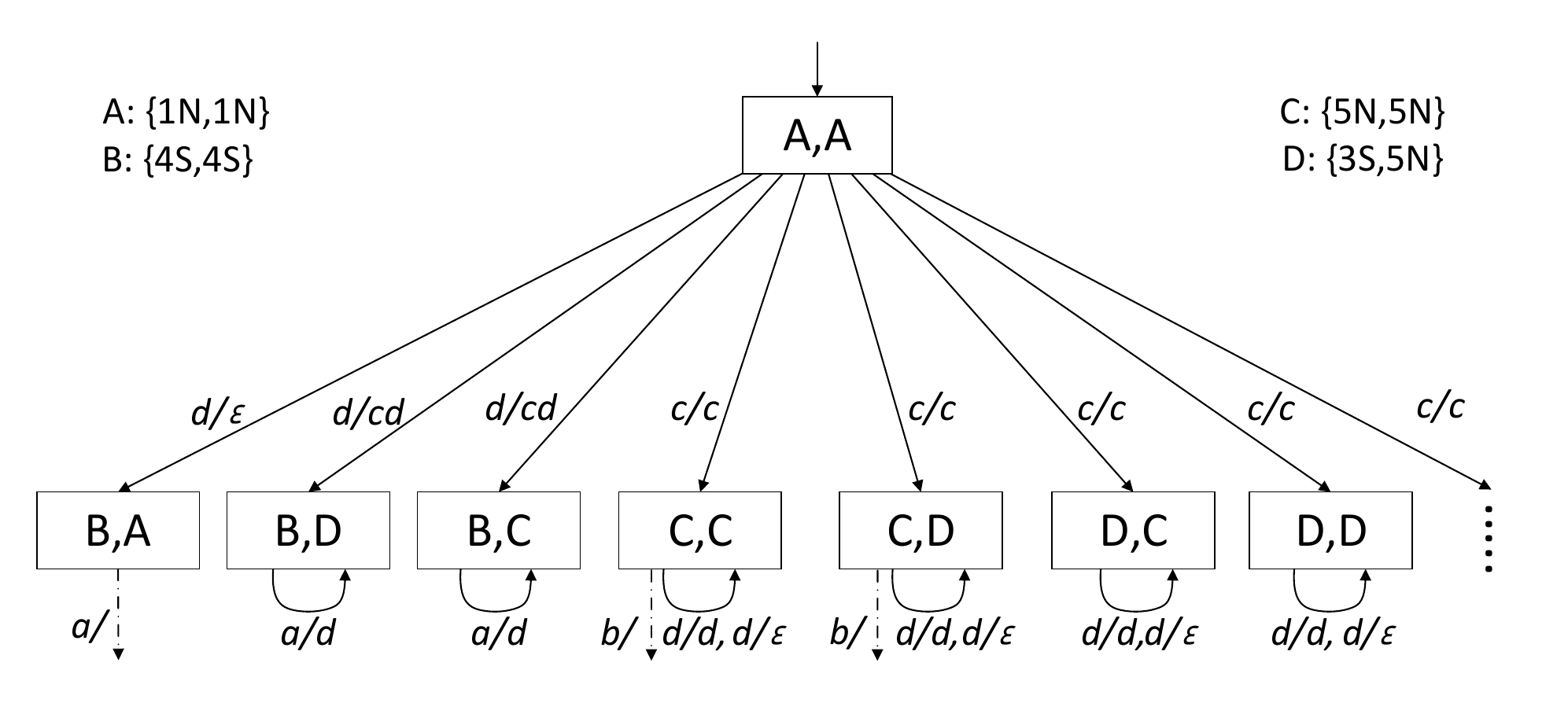}
	}
	\caption{$E$-verifier $V_E$.}
	\label{Everifier}
\end{figure}
\begin{eg}
	Consider the NFA $G$ in Fig. \ref{G} with verifier $G_v$ and defensive verifier ${G_D}$ as shown in Fig. \ref{GvGD}. 
	The $E$-verifier of $G$ is constructed following Definition \ref{RdefinitionKverifier} (partially depicted in Fig. \ref{Everifier} due to space limitation).
	For convenience, partial states have been renamed as $A,B,C$ and $D$, as shown in the figure.
	Note that we draw dotted arrows corresponding to unfeasible defensive actions in response to events $a$ and $b$, denoted by $a/$ and $b/$, at states $(B,A)$, $(C,C)$, and $(C,D)$.
	
	First, we check whether the necessary condition in Theorem~1 is violated via Algorithm~1.
	For instance, the set of co-relative states\footnote{We only consider part of the $E$-verifier in Fig. \ref{Everifier} for the sake of illustration.} for state $4S$ in $G_o$ is $X_{E_{4S}}=\lbrace (B,A), (B,D), (B,C)\rbrace$. 
	We only need to check if a defensive action for $a\in E_o$ is possible among $(B,A), (B,D)$, and $(B,C)$ (because $a$ is the only feasible event at state $4S$).
	Since there are subsequent defensive actions from $(B,D)$ and $(B,C)$, the set of co-relative states for state $4S$ is legal. 
	Similarly, we can check the set of co-relative states for state $5N$ in $G_o$, i.e., $X_{E_{5N}}=\lbrace (C,C), (C,D), (D,C), (D,D)\rbrace$.
	It turns out that there exists no defensive action for $b$.
	Thus, the set of co-relative states for state $5N$ is illegal and Algorithm~1 stops.
	We know that the defensive function is not $C$-enforcing.
	
	One can also check whether the sufficient condition in Theorem 2 is violated via the reduced $E$-verifier $RV_E$ constructed by following Algorithm 2 (in this case, we already know that the sufficient condition will be violated since the necessary condition is not satisfied).
	Since state $(B,A)$ has no continuation (there is no defensive action for event $a$), it is marked as problematic.
	For states $(C,C)$ and $(C,D)$, they are marked as problematic since there is no defensive action for event $b$.
	For state $(A,A)$, even though $f_E((A,A),d/\varepsilon)$ leads to a problematic state $(B,A)$, there exists a feasible defensive action $d/cd$ that leads to states $(B,C)$ and $(B,D)$; even though $f_E((A,A),c/c)$ leads to problematic states $(C,C)$ and $(C,D)$, there exists a feasible defensive action $c/c$ that leads to states $(D,C)$, $(D,D)$, $(E,C)$, and $(E,D)$.
	Therefore, state $(A,A)$ is not marked as problematic.	
	This part of the reduced $E$-verifier $RV_E$ can be directly obtained by removing states $(B,A)$, $(C,C)$, and $(C,D)$.
	Thus, one cannot conclude that the defensive function is $C$-enforcing since the sufficient condition in Theorem 2 is violated (there is no state composed with $C$ and $F$ in $RV_E$).
	The important aspect of the above discussions is that both conditions (sufficient and necessary) can be checked with polynomial complexity.
	\hfill$\diamond$
\end{eg}

\begin{remark}
	Consider a system $G=(X,E,f,x_0)$, where the number of states is $|X|$ and that of secret events is $|E_S|=1$.
	Recall that the state space of the verifier $G_v$ has $|X_v|=(|X|\times 2^{|E_S|})^2=4 |X|^2$ states; identical complexity is exhibited by the defensive verifier, i.e., $|X_D|=4 |X|^2$.
	Then, the number of states of $RV_E$ is bounded by $|X_{ER}|=|X_v|\times|X_D|=16 |X|^4$.
	Overall, the structural complexity of $RV_E$ is $O(|X|^4)$.
\end{remark}
\begin{remark}
	Recall that a defensive function is said to be $C$-enforcing if, regardless of what event is generated by a given system, it can manipulate observations and output a sequence that does not reveal the occurrences of secret events.
	When the defensive function is unconstrained (i.e., each event $t$ in $E_o$ that is produced by the system can be replaced with any other event in $E_o$ ($E_R=E_o\dot{\cup}\lbrace\varepsilon\rbrace$), or any event in $E_o$ can be inserted before the output of event $t$ ($E_I=E_o$)),
	one can also use a verifier construction to find an
	observable sequence including no secret event so as to verify $C$-enforceability in polynomial time.
	Note that the existence of such a sequence can be found via the breadth-first search \cite{DFS} on the verifier with complexity $O(|X_v|\times|E_o|)$.
\end{remark}

\section{Conclusions}

This paper deals with the problem of event concealment for concealing secret events in a system modeled as an NFA under partial observation.
We formally define concealable and unconcealable secret events.
A necessary and sufficient condition is given to verify whether the system is concealable or not.
If concealability of the system does not hold, then we deal with the problem of concealability enforcement.
The notion of $C$-enforceability characterizes whether an external defensive function has the capability to use an obfuscation strategy that manipulates the outputs generated by the system such that the occurrences of unconcealable events cannot be revealed.
It is worth mentioning that the focus of this paper is on the use of reduced complexity constructions (with polynomial complexity) to provide one necessary condition and one sufficient condition for $C$-enforceability. 

\bibliographystyle{IEEEtran}
\bibliography{concealability}

\end{document}